\pdfoutput=1
\documentclass[a4paper,11pt]{article}
\usepackage{fullpage}
\usepackage[T1]{fontenc}
\usepackage[utf8]{inputenc}
\usepackage{latexsym,amssymb,amsmath,amsthm,graphicx,enumerate,mathtools}
\usepackage{cite,braket}
\usepackage{wasysym}
\usepackage{thmtools}
\usepackage{thm-restate}
\usepackage[linesnumbered,ruled,vlined]{algorithm2e}
\usepackage{algpseudocode}
\usepackage{algorithmicx}
\usepackage{algcompatible}
\usepackage[hidelinks]{hyperref}

\theoremstyle{plain}
\newtheorem{theorem}{Theorem}

\newtheorem{lemma}[theorem]{Lemma}
\newtheorem{corollary}[theorem]{Corollary}

\theoremstyle{definition}
\newtheorem{definition}[theorem]{Definition}
\newtheorem{example}[theorem]{Example}

\renewcommand{\S}{\mathcal{S}}
\newcommand{\R}{\mathcal{R}}

\newcommand{\C}{\mathcal{C}}
\providecommand{\cl}{}
\renewcommand{\cl}{\operatorname{cl}}
\newcommand{\Blow}{B_{\mathrm{low}}}
\newcommand{\Bhigh}{B_{\mathrm{high}}}
\newcommand{\gap}{{\sf gap}}
\newcommand{\best}{\Delta}

\title{Stable Matchings with Minimum Utility Gap}
\author{
Yao Sheng\thanks{Department of Mathematical and Computing Science, School of Computing, Institute of Science Tokyo, Tokyo 152-8552, Japan. Email: \texttt{sheng.y.2ab1@m.isct.ac.jp}}
\and
Yu Yokoi\thanks{Department of Mathematical and Computing Science, School of Computing, Institute of Science Tokyo, Tokyo 152-8552, Japan. Email: \texttt{yokoi@comp.isct.ac.jp}}
}
\date{}

\begin{document}
\maketitle
\begin{abstract}
We introduce the {\em Stable Matching Problem with Minimum Utility Gap}, which seeks a stable matching in which the utilities received by individual agents are as balanced as possible. Our framework can handle many-to-many matchings and general utility functions on partner sets that are consistent with the agents' preferences. We consider two measures for comparing agents' utilities: the difference between the maximum and minimum utilities, and their ratio.

We provide a polynomial-time algorithm for both versions. The algorithm exploits the rotation-poset representation of the set of stable matchings and, in particular, the fact that the rotations affecting each agent form a chain in this poset. To position our result, we also clarify its relation to existing frameworks: we show that our objectives are not captured by the recent minimum-cut representability framework, while identifying a special case that admits a submodular function minimization interpretation.
\end{abstract}

\medskip

\section{Introduction}
\label{sec:introduction}

Since the seminal work of Gale and Shapley \cite{GaleShapley1962}, the theory of stable matching has grown into a rich research area, with deep mathematical structure and broad applications; see, e.g., \cite{RothSotomayor1990,Manlove2013}.
A matching in a two-sided market is {\em stable} if there is no pair
of agents on opposite sides who are not matched to each other and
can mutually benefit by forming a partnership. The Gale--Shapley deferred-acceptance algorithm guarantees the existence of a stable matching.
Throughout this paper, we denote the two sides of the market by $S$ and $L$ (students and laboratories), and call each element of $S\cup L$ an agent.

Even in the classical one-to-one setting, known as the {\em stable marriage} problem, the number of stable matchings can be exponential in the input size \cite{IrvingLeather1986,GusfieldIrving1989}.
The stable matching returned by the Gale--Shapley algorithm is optimal for one side, but pessimal for the other side \cite{Knuth1976Mariages,GusfieldIrving1989}.
It is then natural to ask whether we can find a stable matching that is quantitatively optimal with respect to a given objective function.

Indeed, various stable matching optimization problems have been studied since the 1980s, especially for the marriage setting.
The {\em egalitarian stable matching} problem asks for a stable matching $M$ minimizing $\sum_{a\in S\cup L} \operatorname{rank}_a(M(a))$, where $\operatorname{rank}_a(M(a))\in\{1,2,\dots\}$ is the rank of the assigned partner of agent $a$; a smaller rank is better \cite{IrvingLeatherGusfield1987,GusfieldIrving1989}.
The {\em minimum-regret stable matching} problem asks for a stable matching minimizing $\max_{a\in S\cup L} \operatorname{rank}_a(M(a))$ \cite{Gusfield1987,GusfieldIrving1989}.
Both of these problems are tractable.
More generally, Irving--Leather--Gusfield \cite{IrvingLeatherGusfield1987} showed that, for any weight function on acceptable pairs, a stable matching of minimum or maximum total weight can be found efficiently.\footnote{The minimum-regret stable matching problem can be reduced to the minimum-weight case by using exponentially separated weights.}
This tractability is based on the distributive lattice structure of stable matchings and its compact representation by a rotation poset, which allows the problem to be reduced to a minimum-cut problem \cite{GusfieldIrving1989}.
Recent work of Faenza--Foussoul--He \cite{FaenzaFoussoulHe2025} gives a characterization of stable matching optimization problems, not necessarily linear, that can be represented as a minimum-cut problem.

On the other hand, not all natural objective functions over stable matchings are tractable.
For example, the {\em sex-equal stable marriage} problem asks for a stable matching $M$ minimizing $|\sum_{s\in S}\operatorname{rank}_s(M(s))-\sum_{\ell\in L}\operatorname{rank}_\ell(M(\ell))|$, and is NP-hard \cite{GusfieldIrving1989,Kato1993}.
The {\em balanced stable marriage} problem asks for a stable matching minimizing the maximum of $\sum_{s\in S}\operatorname{rank}_s(M(s))$ and $\sum_{\ell\in L}\operatorname{rank}_\ell(M(\ell))$, and is also NP-hard \cite{Feder1995,GuptaRoySaurabhZehavi2019}.
These objectives are motivated by fairness between the two sides of the market.

In this paper, we study a different type of fairness objective, where we compare the utilities received by individual agents.

\subsection*{Our Contribution}
We propose a new stable matching optimization problem, which we call the {\em Stable Matching Problem with Minimum Utility Gap}.
Like the sex-equal stable matching problem, our objective is based on a gap; however, instead of comparing aggregate quantities of the two sides, we compare utilities received by individual agents.
For example, in the one-to-one setting, our framework contains, as a special case, the problem of minimizing
$\max_{a,a'\in S\cup L}
\left|
\operatorname{rank}_a(M(a))-\operatorname{rank}_{a'}(M(a'))
\right|$.
We formulate the problem in a more general many-to-many setting, where each agent $a\in S\cup L$ has a capacity $q_a$ in addition to a strict preference list over agents on the opposite side.
We are also given a nonempty subset $A\subseteq S\cup L$ of agents whose utilities we care about, and for each $a\in A$, a value function $v_a$ that assigns a value to each possible partner set of $a$.
We assume that each $v_a$ is consistent with the preference of $a$, in the sense that a better partner set with respect to the preference has no smaller value; see condition~($\star$) in Section~\ref{sec:problem-setting} for the precise definition.
Our first objective is the difference version: we aim to find a stable matching $M$ minimizing
\[
\max_{a\in A} v_a(M(a))-\min_{a\in A}v_a(M(a)),
\]
where $M(a)$ denotes the set of partners assigned to $a$ in $M$.
We also consider the ratio version, where, assuming positivity of all relevant function values, we aim to minimize
\[
\frac{\max_{a\in A} v_a(M(a))}{\min_{a\in A}v_a(M(a))}.
\]
We show that both versions can be solved efficiently by the same algorithmic framework, with only a slight modification.%
\footnote{For positive utilities, ratio minimization is equivalent to difference minimization for log-transformed utilities. We nevertheless treat the ratio objective directly to avoid evaluating logarithms. Since logarithmic and exponential transformations do not generally preserve submodularity, we give separate submodularity arguments for the two objectives in Section~\ref{sec:comparison}.}
The rank-gap objective mentioned above is obtained from the difference version by taking $A=S\cup L$ and $v_a(M(a))=-\operatorname{rank}_a(M(a))$.

As with other stable matching optimization problems~\cite{ManloveIrvingIwamaMiyazakiMorita2002}, our problems become NP-hard under weak stability when ties are allowed. The proofs appear in Appendix~\ref{app:hardness}.

\paragraph*{Algorithm Overview.}
Our algorithm exploits the distributive lattice structure of stable matchings. As mentioned earlier, the stable matching lattice admits a polynomial-size representation by a rotation poset, which is true even in the many-to-many setting \cite{BansalAgrawalMalhotra2007,EirinakisMagosMourtosMiliotis2012}.

Using this representation, our problem is reduced to finding a closed set of the rotation poset corresponding to a stable matching with minimum utility gap.
Each rotation can be seen as an operation that transforms one stable matching into another.
A key observation used in our algorithm is that, for any agent $a\in S\cup L$, the set of rotations that change the assignment of $a$ forms a chain of length at most $n=\max\{|S|,|L|\}$; we denote this chain by $\C(a)$.
Furthermore, the assignment of $a$ in the stable matching corresponding to a closed set $X$ depends only on the prefix $X\cap\C(a)$ of the chain $\C(a)$, where this prefix may be empty.
The utility value received by $a$ is monotone along this chain: it is non-increasing if $a\in S$ and non-decreasing if $a\in L$.

Exploiting these structural properties, we first give a procedure that, for a given interval $[\Blow,\Bhigh]$, tests whether there exists a stable matching $M$ such that
$\Blow\leq v_a(M(a))\leq \Bhigh$ for every $a\in A$.
The interval constraints are translated into forced and forbidden rotations, and feasibility is checked by computing a downward closure in the rotation digraph.
The above observations also imply that the set of utility values that can appear for agents in $A$ has size $O(n^2)$.
By applying a sliding-window search over these values, we find an optimal feasible interval using only $O(n^2)$ feasibility tests.
The resulting running time is $O(n^4+n^2T_v)$, where $T_v$ is the time needed to evaluate one value $v_a(X)$.

\paragraph*{Comparison with Existing Frameworks.}
To position our result, we investigate its relation to existing frameworks.
First, we show that neither the difference nor the ratio version of our objective is captured by the minimum-cut representability framework of Faenza--Foussoul--He \cite{FaenzaFoussoulHe2025}, which includes various other tractable stable matching optimization problems.
We give a concrete instance in which the characterization condition of \cite{FaenzaFoussoulHe2025} fails for our objectives.

Second, we show that, when $A$ lies entirely on one side, i.e., $A\subseteq S$ or $A\subseteq L$, both of our objective functions have a submodularity interpretation.
More precisely, by regarding the objective as a function on the family of closed subsets of the rotation poset, we obtain a submodular function over this ring family. Such a function can be efficiently minimized \cite{Schrijver2000,IwataFleischerFujishige2001}.
This fact may be of independent interest, while we also show that such an interpretation does not extend to the general case where $A$ contains agents from both $S$ and $L$.

These observations highlight the need for an algorithm tailored to the present objective.

\subsubsection*{Further Related Work}

The rotation-poset approach to stable matchings has a long history.
For the marriage setting, it was developed by Irving--Leather--Gusfield \cite{IrvingLeatherGusfield1987}, Gusfield \cite{Gusfield1987}, and Gusfield--Irving \cite{GusfieldIrving1989}.
For many-to-many matchings, Ba\"{\i}ou--Balinski \cite{BaiouBalinski2000ManyToMany} established fundamental lattice properties, and Bansal--Agrawal--Malhotra \cite{BansalAgrawalMalhotra2007} developed rotation-based algorithms for finding optimal stable matchings; this line was further developed by Eirinakis--Magos--Mourtos--Miliotis \cite{EirinakisMagosMourtosMiliotis2012}.
More recent work studies compact representations and rotation-like structures for models with choice functions, including the works of Faenza--Zhang \cite{FaenzaZhang2021} and Karzanov \cite{Karzanov2024}.

An alternative approach to linear optimization over bipartite stable matchings is developed by Fleiner--Frank--Kir{\'a}ly~\cite{FFK24}. They construct a directed graph in which stable matchings are represented as certain cuts, without using the rotation poset. This representation also supports packing and covering problems, including finding a prescribed number of pairwise edge-disjoint stable matchings with minimum total cost.

There is also a broad literature on fairness-oriented optimization in two-sided matching markets. Besides the work discussed above, there are studies on leximin fairness~\cite{NarangBiswasNarahari2022} and Nash social welfare~\cite{JainVaish2024}.
Median and generalized median stable matchings are connected to the geometry of stable matchings and fairness considerations \cite{TeoSethuraman1998,SethuramanTeoQian2006,Cheng2010}.

\section{Problem Setting}
\label{sec:problem-setting}
An instance of our problem, the {\em Stable Matching Problem with Minimum Utility Gap}, consists of an instance
$I=(S,L,\{\succ_a\}_{a\in S\cup L},\{q_a\}_{a\in S\cup L})$ of the standard many-to-many matching market, together with a nonempty agent subset $A\subseteq S\cup L$ and value functions $\{v_a\}_{a\in A}$.
Here, $S$ is a set of students, $L$ is a set of laboratories, and each element of $S\cup L$ is called an agent.
We write $n=\max\{|S|,|L|\}$.

Each agent has a strict preference list over a subset of agents on the opposite side, and $x\succ_a y$ means that agent $a$ strictly prefers $x$ to $y$.
The positive integer $q_a$ denotes the capacity of agent $a$.
If agent $a$ appears in agent $x$'s preference list, then $a$ is said to be {\em acceptable} to $x$.
Since we consider only mutually acceptable pairs in matchings, we assume without loss of generality that acceptability is mutual.

For each $a\in A$, the value function $v_a$ assigns a value to each subset of acceptable agents on the opposite side.
Throughout the paper, we assume that $v_a$ is {\em consistent} with $\succ_a$, meaning that it satisfies the following condition:

\begin{itemize}
    \item[($\star$)] For two sets $X$ and $Y$ of agents acceptable to $a$, if $|X|=|Y|$ and $x\succ_a y$ holds for every $x\in X$ and every $y\in Y\setminus X$, then $v_a(X)\geq v_a(Y)$.
\end{itemize}

\noindent Simple examples satisfying this condition are the following.
\begin{description}
    \item[Total Utility:] There is a real-valued weight function $w_a$ on acceptable agents such that $x\succ_a y\implies w_a(x)\geq w_a(y)$, and the value of partner set $X$ is $v_a(X)=\sum_{x\in X}w_a(x)$.
    \item[Average Value:] There is a weight function $w_a$ as above, and $v_a(X)=\frac{1}{|X|}\sum_{x\in X}w_a(x)$ for each nonempty partner set $X$, while $v_a(\emptyset)=0$.
\end{description}

Any function satisfying ($\star$) can be handled in our framework, provided that its value can be accessed by a polynomial-time oracle. In the two examples above, the weight functions $w_a$ may instead be given explicitly.

Let $E\subseteq S\times L$ be the set of all acceptable pairs.
For $M\subseteq E$ and an agent $a\in S\cup L$, let $M(a)$ denote the set of partners of $a$ in $M$; for example, $M(\ell)=\set{s\in S\mid (s,\ell)\in M}$ for $\ell\in L$.
A set $M\subseteq E$ is a {\em matching} if $|M(a)|\leq q_a$ for every $a\in S\cup L$.
For a matching $M$, a pair $(s,\ell)\in E\setminus M$ is a {\em blocking pair} if the following two conditions hold:
\begin{itemize}
    \item student $s$ satisfies $|M(s)|<q_s$, or $\ell\succ_s \ell'$ holds for some $\ell'\in M(s)$; and
    \item laboratory $\ell$ satisfies $|M(\ell)|<q_\ell$, or $s\succ_\ell s'$ holds for some $s'\in M(\ell)$.
\end{itemize}
A matching is {\em stable} if it admits no blocking pair.
Stability depends only on $I$, and we denote the set of all stable matchings in $I$ by $\mathcal S(I)$.

It is known that $|\mathcal S(I)|$ can be exponential in $n$. We seek a stable matching that is fair for the agents in the specified set $A$.
For a stable matching $M$ and an agent $a\in A$, the utility of $a$ is $v_a(M(a))$.
We consider two objective functions. The difference version is
\[
\Delta_{\rm diff}(M)\coloneqq \max_{a\in A} v_a(M(a))-\min_{a\in A} v_a(M(a)).
\]
In the ratio version below, we assume that all relevant function values are positive:
\[
\Delta_{\rm rat}(M)\coloneqq \frac{\max_{a\in A} v_a(M(a))}{\min_{a\in A} v_a(M(a))}.
\]
The {\em Stable Matching Problem with Minimum Utility Gap} asks, given $I$, $A$, and $\{v_a\}_{a\in A}$, to find a stable matching $M\in\mathcal S(I)$ minimizing either $\Delta_{\rm diff}$ or $\Delta_{\rm rat}$.
As we will see, both versions can be solved by the same framework, with only the final comparison rule changed.
\section{Preliminaries}
\label{sec:preliminaries}

Since our algorithm depends on known structural properties of the set of stable matchings, we review them in this section. Let $I=(S,L,\{\succ_a\}_{a\in S\cup L},\{q_a\}_{a\in S\cup L})$ be a many-to-many stable matching instance as described in the previous section, and let $n=\max\{|S|,|L|\}$.

We list some standard facts that will be used in our analysis. These facts are shown for the one-to-one or many-to-one cases by Gusfield--Irving \cite{GusfieldIrving1989}, and for the many-to-many case by Ba\"{\i}ou--Balinski \cite{BaiouBalinski2000ManyToMany}, Bansal--Agrawal--Malhotra \cite{BansalAgrawalMalhotra2007}, and Eirinakis--Magos--Mourtos--Miliotis \cite{EirinakisMagosMourtosMiliotis2012}.

\begin{theorem}[Rural Hospitals Theorem {\cite[Theorem~2 and Corollary on p.~5]{BaiouBalinski2000ManyToMany}}]
\label{thm:RHT}
For every two stable matchings $M, M'\in \S(I)$, each agent $a\in S\cup L$ satisfies $|M(a)|=|M'(a)|$. Moreover, if $|M(a)|<q_a$ in some stable matching $M$, then $M(a)=M'(a)$ for all stable matchings $M'$.
\end{theorem}

For an agent $a\in S\cup L$ and any two sets $X$ and $Y$ of agents acceptable to $a$, we write $X\succeq_a Y$ if $|X|=|Y|$ and $x\succ_a y$ holds for every $x\in X$ and every $y\in Y\setminus X$ (possibly $Y\setminus X=\emptyset$, i.e., $X=Y$). We write $X\succ_a Y$ if $X\succeq_a Y$ and $X\neq Y$.

\begin{theorem}[Comparability of assigned sets {\cite[Theorem~3]{BaiouBalinski2000ManyToMany}}]
\label{thm:comparability}
For every two stable matchings $M,M'\in \S(I)$ and every agent $a\in S\cup L$, we have $M(a)\succeq_a M'(a)$ or $M'(a)\succeq_a M(a)$.
\end{theorem}

For stable matchings $M$ and $M'$, we write $M\succeq_S M'$ if $M(s)\succeq_s M'(s)$ for every student $s\in S$. Thus, larger elements with respect to $\succeq_S$ are better for students.

\begin{theorem}[Distributive lattice structure {\cite[Theorem~4]{BaiouBalinski2000ManyToMany}}]
\label{thm:lattice}
The partially ordered set $(\mathcal S(I),\succeq_S)$ forms a distributive lattice. For two stable matchings $M,M'\in \mathcal S(I)$, the join $M\vee M'$ is obtained by assigning to each student $s\in S$ the better of the two sets $M(s)$ and $M'(s)$ with respect to $\succeq_s$, while the meet $M\wedge M'$ is obtained by assigning to each student $s\in S$ the worse of the two sets.
\end{theorem}

Theorem~\ref{thm:lattice} in particular implies the existence of a unique maximum element and a unique minimum element in $(\mathcal S(I),\succeq_S)$.
Let $M_0$ and $M_z$ denote the maximum and minimum stable matchings, respectively. We call $M_0$ the student-optimal stable matching and $M_z$ the student-pessimal stable matching.

By Birkhoff's representation theorem, every finite distributive lattice is isomorphic to the lattice of closed sets of a partially ordered set. In the case of the stable matching lattice, this poset is known to be represented by a partial order on rotations, which are defined below.
We follow the definitions in \cite[Definition~2]{BansalAgrawalMalhotra2007} and \cite[Definition~4]{EirinakisMagosMourtosMiliotis2012} with a slight modification.\footnote{In \cite{BansalAgrawalMalhotra2007}, this object is called a {\em meta-rotation}, to be distinguished from rotations in the one-to-one matching case. For simplicity, we use the term {\em rotation}.}

\begin{definition}[Rotation]
\label{def:rotation}
A {\em rotation} is an ordered list $\rho=((s_0,\ell_0),\ldots,(s_{t-1},\ell_{t-1}))$, where $2\leq t\leq n$, satisfying the following conditions for some stable matching $M$:
\begin{itemize}
    \item each pair $(s_i,\ell_i)$ belongs to $M$, and any agent appears at most once in $\rho$;
    \item $|M(\ell_i)|=q_{\ell_i}$, and $s_i$ is the worst student for $\ell_i$ among those in $M(\ell_i)$;
    \item $\ell_i$ is the best laboratory for $s_{i+1}$ among acceptable laboratories $\ell\notin M(s_{i+1})$ such that either $|M(\ell)|<q_\ell$ or $s_{i+1}\succ_\ell s$ for some $s\in M(\ell)$,
\end{itemize}
where indices are taken modulo $t$. When these conditions hold for a stable matching $M$, we say that $\rho$ is {\em exposed} in $M$.
\end{definition}

\begin{definition}[Elimination of a rotation]
\label{def:rotation-elimination}
Let $\rho=((s_0,\ell_0),\ldots,(s_{t-1},\ell_{t-1}))$ be a rotation. Define $\rho^-=\{(s_0,\ell_0),\ldots,(s_{t-1},\ell_{t-1})\}$ and $\rho^+=\{(s_1,\ell_0),(s_2,\ell_1),\ldots,(s_0,\ell_{t-1})\}$, where indices are taken modulo $t$. If $\rho$ is exposed in a stable matching $M$, then eliminating $\rho$ from $M$ means replacing $M$ by $M/\rho\coloneqq (M\setminus\rho^-)\cup\rho^+$.
\end{definition}

The matching $M/\rho$ is stable whenever $\rho$ is exposed in $M$ (cf.~\cite[Lemma~1]{EirinakisMagosMourtosMiliotis2012}); see Lemma~\ref{lem:rotation-stability} in Appendix~\ref{app:omitted} for a direct proof. 
It also follows from the definition of a rotation that $M\succeq_S M/\rho$ (see Lemma~\ref{lem:one-step-monotonicity} for the proof). Thus, repeated application of rotation elimination yields a sequence of stable matchings that monotonically worsens for students. In fact, all stable matchings can be obtained in this way.

\begin{theorem}[Traversal by rotations {\cite[Theorem~2]{BansalAgrawalMalhotra2007}}]
\label{thm:traversal}
Every stable matching can be obtained from the student-optimal stable matching $M_0$ by repeatedly eliminating a rotation that is exposed at the current stable matching.
\end{theorem}

Let $\mathcal R(I)$ be the set of all rotations of $I$. As an immediate consequence of \cite[Lemma 6]{BansalAgrawalMalhotra2007}, we have $|\mathcal R(I)|\leq |S||L|$, and hence $|\mathcal R(I)|=O(n^2)$.
We now define the precedence order on $\mathcal R(I)$ to obtain the rotation poset.

\begin{definition}[Rotation poset]
\label{def:rotation-poset}
For two rotations $\rho,\rho'\in\mathcal R(I)$, we write $\rho\lhd\rho'$ if the following holds: whenever $\rho'$ is exposed in a stable matching $M$, every sequence of successive rotation eliminations from $M_0$ to $M$ contains $\rho$. We write $\rho\unlhd\rho'$ if $\rho\lhd\rho'$ or $\rho=\rho'$.
\end{definition}

A set $X\subseteq\mathcal R(I)$ of rotations is called {\em closed} if $\rho'\unlhd\rho\in X$ implies $\rho'\in X$.

\begin{theorem}[Rotation-poset representation {\cite[Theorem~3]{BansalAgrawalMalhotra2007}}]
\label{thm:isomorphism}
There is a one-to-one correspondence between $\mathcal S(I)$ and the closed subsets of the rotation poset $(\mathcal R(I),\unlhd)$. More precisely, for every closed subset $X\subseteq\mathcal R(I)$, eliminating all rotations in $X$ from $M_0$, in any order consistent with $\unlhd$, yields a stable matching, and this correspondence is bijective.
\end{theorem}

One way to represent $(\mathcal R(I),\unlhd)$ is to use a directed graph. In the literature, rotation digraphs are usually defined by explicit constructions. For our purpose, however, it is sufficient to use the following property as the definition.

\begin{definition}[Rotation digraph]
\label{def:rotation-digraph}
A rotation digraph of an instance $I$ is a directed acyclic graph $G=(\mathcal R(I),A_G)$ such that, for every $\rho,\rho'\in\mathcal R(I)$, we have $\rho\unlhd\rho'$ if and only if $\rho'$ is reachable from $\rho$ in $G$. Equivalently, the arc set $A_G$ contains the arcs of the Hasse diagram of $(\mathcal R(I),\unlhd)$ and is contained in the transitive closure of this Hasse diagram.
\end{definition}

Bansal--Agrawal--Malhotra \cite{BansalAgrawalMalhotra2007} gave an efficient construction of a rotation digraph in the many-to-many setting, which was later sped up by Eirinakis--Magos--Mourtos--Miliotis \cite{EirinakisMagosMourtosMiliotis2012}.

\begin{theorem}[Rotation digraph construction {\cite[Lemma~14]{EirinakisMagosMourtosMiliotis2012}}]
\label{thm:rotation-digraph-construction}
Given a many-to-many stable matching instance $I$, a rotation digraph $G=(\mathcal R(I),A_G)$ can be constructed in $O(n^2)$ time, where $n=\max\{|S|,|L|\}$. The constructed graph has $O(n^2)$ arcs.
\end{theorem}

\section{Proposed Algorithm}
\label{sec:algorithm}

This section presents our algorithmic results. Before describing our algorithm, we first prepare some additional structural properties that we exploit.

\subsection{Structural Observations}

For any rotation $\rho\in\mathcal R(I)$ and any student $s\in S$ involved in $\rho$, we write $\rho^-(s)$ and $\rho^+(s)$ for the unique laboratories such that $(s,\rho^-(s))\in\rho^-$ and $(s,\rho^+(s))\in\rho^+$. Similarly, for a laboratory $\ell\in L$ involved in $\rho$, we write $\rho^-(\ell)$ and $\rho^+(\ell)$ for the unique students such that $(\rho^-(\ell),\ell)\in\rho^-$ and $(\rho^+(\ell),\ell)\in\rho^+$.
For a set $X$ and an element $y$, abusing notation slightly, we write $X\succ_a y$ to mean that $x\succ_a y$ holds for every $x\in X$.

The following property follows directly from the definition of rotations.

\begin{lemma}
\label{lem:one-step-monotonicity}
Let $\rho$ be a rotation exposed in a stable matching $M$.
\begin{description}
    \item[\normalfont (i)] For any $s\in S$ involved in $\rho$, we have $M(s)\succ_s\rho^+(s)$. In particular, $\rho^-(s)\succ_s\rho^+(s)$.
    \item[\normalfont (ii)] For any $\ell\in L$ involved in $\rho$, we have $(M/\rho)(\ell)\succ_\ell\rho^-(\ell)$. In particular, $\rho^+(\ell)\succ_\ell\rho^-(\ell)$.
\end{description}
\end{lemma}

\begin{proof}
By the definition of a rotation, we have $\rho^+(s)\notin M(s)$, and $\rho^+(s)$ prefers $s$ to its worst assignee in $M$. If $s$ preferred $\rho^+(s)$ to some laboratory in $M(s)$, then $(s,\rho^+(s))$ would block $M$. Hence, $s$ prefers every element in $M(s)$ over $\rho^+(s)$, i.e., $M(s)\succ_s\rho^+(s)$.

By the definition of a rotation, $\rho^-(\ell)$ is the worst student in $M(\ell)$, and $\ell$ prefers $\rho^+(\ell)$ to $\rho^-(\ell)$. Since $(M/\rho)(\ell)=M(\ell)\setminus\{\rho^-(\ell)\}\cup\{\rho^+(\ell)\}$, we have $(M/\rho)(\ell)\succ_\ell\rho^-(\ell)$.
\end{proof}

The following lemma consists of three statements, the first of which was shown as Lemma~6 by Bansal--Agrawal--Malhotra \cite{BansalAgrawalMalhotra2007}. We provide proofs of the latter two in Appendix~\ref{app:omitted}.

\begin{restatable}{lemma}{uniquenesslem}
\label{lem:uniqueness}
The following statements hold.
\begin{description}
\item[\normalfont (i)] For every pair $(s,\ell)$, there is at most one rotation $\rho$ such that $(s,\ell)\in\rho^-$.
\item[\normalfont (ii)] For every pair $(s,\ell)$, there is at most one rotation $\rho$ such that $(s,\ell)\in\rho^+$.
\item[\normalfont (iii)] There is no rotation $\rho$ such that $M_z\cap\rho^-\neq\emptyset$ or $M_0\cap\rho^+\neq\emptyset$.
\end{description}
\end{restatable}

The following property is known in the one-to-one case \cite[Claim~4.6]{karlin2018simply} and was also shown in the many-to-one case \cite[Lemma~17]{FaenzaFoussoulHe2025}. We include a proof for the many-to-many case.

\begin{lemma}
\label{lem:chain-agent}
For any agent $a\in S\cup L$, all rotations involving agent $a$ are pairwise comparable in the poset $(\mathcal R(I),\unlhd)$.
\end{lemma}

\begin{proof}
We first consider a student $s$. Let $\rho_1$ and $\rho_2$ be two distinct rotations involving $s$, and write $\ell_j=\rho_j^+(s)$ for $j=1,2$. By Lemma~\ref{lem:uniqueness}(ii), we have $\ell_1\neq \ell_2$.
Suppose that $\rho_1\ntriangleleft\rho_2$.
Then there is a linear extension of $(\R(I),\unlhd)$ in which
$\rho_2$ precedes $\rho_1$.
By Theorem~\ref{thm:isomorphism} applied to $X=\R(I)$, eliminating all
rotations in this order is valid.
Taking the prefix ending immediately before the elimination of
$\rho_1$, we obtain a sequence
$M_0\succeq_S M_1\succeq_S\cdots\succeq_S M_k$ of stable matchings
such that, for some $i<k$, the rotation $\rho_2$ is exposed in $M_i$ and
$M_{i+1}=M_i/\rho_2$, while $\rho_1$ is exposed in $M_k$.
From $M_{i+1}=M_i/\rho_2$, we obtain $\ell_2=\rho^+_2(s)\in M_{i+1}(s)$.
Since $\rho_1$ is exposed in $M_k$, Lemma~\ref{lem:one-step-monotonicity}(i) implies $M_k(s)\succ_s\rho_1^+(s)=\ell_1$.
Since $M_{i+1}(s)\succeq_s M_k(s)$, we have $\ell_2\in M_{i+1}(s)\succeq_s M_k(s)\succ_s \ell_1$, which implies $\ell_2\succ_s\ell_1$.
By the same argument, $\rho_2\ntriangleleft\rho_1$ implies $\ell_1\succ_s\ell_2$. Thus, it is impossible to have both $\rho_1\ntriangleleft\rho_2$ and $\rho_2\ntriangleleft\rho_1$.

We next consider a laboratory $\ell$. Let $\rho_1$ and $\rho_2$ be two distinct rotations involving $\ell$, and write $s_j=\rho_j^-(\ell)$ for $j=1,2$. By Lemma~\ref{lem:uniqueness}(i), we have $s_1\neq s_2$.
Suppose that $\rho_1\ntriangleleft\rho_2$. By the same argument as above, there exists a sequence of stable matchings $M_0\succeq_S M_1\succeq_S\cdots\succeq_S M_k$ obtained by successive rotation eliminations such that, for some index $i<k$, the rotation $\rho_2$ is exposed in $M_i$ and $M_{i+1}=M_i/\rho_2$, while $\rho_1$ is exposed in $M_k$.
From $M_{i+1}=M_i/\rho_2$, Lemma~\ref{lem:one-step-monotonicity}(ii) implies $M_{i+1}(\ell)\succ_{\ell} \rho_2^-(\ell)=s_2$. This property is preserved under later rotation eliminations involving $\ell$, since each such rotation replaces a currently assigned student with one preferred by $\ell$. Hence $M_k(\ell)\succ_\ell s_2$. Since $\rho_1$ is exposed in $M_k$, we have $s_1=\rho_1^-(\ell)\in M_k(\ell)$, and hence $s_1\succ_\ell s_2$.
By the same argument, $\rho_2\ntriangleleft\rho_1$ implies $s_2\succ_\ell s_1$. Therefore, it is impossible to have both $\rho_1\ntriangleleft\rho_2$ and $\rho_2\ntriangleleft\rho_1$.
\end{proof}

For a poset $(\mathcal R(I),\unlhd)$, a {\em chain} is a set of mutually comparable rotations. Its {\em length} is the number of its
elements; thus, a chain of length $k$ can be written as
$\rho_1\lhd\rho_2\lhd\cdots\lhd\rho_k$ after indexing its elements
appropriately.

\begin{corollary}
\label{cor:agent-chain-length}
For each agent $a\in S\cup L$, the set of rotations involving agent $a$ forms a chain $\mathcal C(a)$ in $(\mathcal R(I),\unlhd)$, and its length is at most $n$.
\end{corollary}

\begin{proof}
The chain property is immediate from Lemma~\ref{lem:chain-agent}. If $a=s\in S$, then distinct rotations involving $s$ have distinct $\rho^+(s)$ by Lemma~\ref{lem:uniqueness}(ii), and hence there are at most $|L|$ of them. If $a=\ell\in L$, then distinct rotations involving $\ell$ have distinct $\rho^-(\ell)$ by Lemma~\ref{lem:uniqueness}(i), and hence there are at most $|S|$ of them. Therefore, the chain length is at most $\max\{|S|,|L|\}=n$.
\end{proof}

For any agent $a\in S\cup L$, we write $\mathcal C(a)$ for the chain of rotations involving $a$.

Recall that, by Theorem~\ref{thm:isomorphism}, for any closed subset $X$ of $(\mathcal R(I),\unlhd)$, eliminating $X$ from $M_0$ defines a stable matching, and this correspondence is bijective. Let $M_X\in\mathcal S(I)$ be the stable matching corresponding to a closed set $X$.
For any agent $a\in S\cup L$, only rotations involving $a$ can change the assignment of $a$, and hence the set $M_X(a)$ depends only on $X\cap\mathcal C(a)$. That is, for closed sets $X,Y\subseteq\mathcal R(I)$, $X\cap\mathcal C(a)=Y\cap\mathcal C(a)$ implies $M_X(a)=M_Y(a)$.

For each $a\in A$ and each rotation $\rho\in\mathcal C(a)$, we define
\[
\hat v_a(\rho)\coloneqq v_a(M_X(a)),
\quad\text{where $X\subseteq\mathcal R(I)$ is any closed set with $\rho=\max_{\unlhd}(X\cap\mathcal C(a))$}.
\]
This is well-defined because (i) a closed set $X$ with the above condition exists, for example, by taking the downward closure of $\rho$, and (ii) for any such closed set $X$, the set $X\cap\mathcal C(a)$ is a prefix of the chain $\mathcal C(a)$, and hence its maximum $\rho$ determines this intersection uniquely.

\begin{lemma}
\label{lem:value-monotonicity}
For any $s\in S\cap A$ and any $\ell\in L\cap A$, we have the following.
\begin{itemize}
    \item $\hat v_s$ is monotone non-increasing on $\mathcal C(s)$; i.e., $\rho\lhd\rho'$ implies $\hat v_s(\rho)\geq \hat v_s(\rho')$.
    \item $\hat v_\ell$ is monotone non-decreasing on $\mathcal C(\ell)$; i.e., $\rho\lhd\rho'$ implies $\hat v_\ell(\rho)\leq \hat v_\ell(\rho')$.
\end{itemize}
\end{lemma}

\begin{proof}
Let $a\in A$. If $\mathcal C(a)=\emptyset$, there is nothing to prove. Otherwise, let $\rho_1\lhd\rho_2\lhd\cdots\lhd\rho_k$ be the chain $\mathcal C(a)$. For $i=1,2,\ldots,k$, let $M_i$ be the stable matching corresponding to the closed set $X_i=\set{\tilde{\rho}\in\mathcal R(I)|\, \tilde{\rho}\, \unlhd \rho_i}$. Then $\hat v_a(\rho_i)=v_a(M_i(a))$. Note that $X_1\subseteq X_2\subseteq\cdots\subseteq X_k$. By Lemma~\ref{lem:one-step-monotonicity}, rotation elimination can only change the assignment of a student to a worse one and that of a laboratory to a better one. Hence, we have $M_1(a)\succeq_a M_2(a)\succeq_a\cdots\succeq_a M_k(a)$ if $a\in S$, and $M_k(a)\succeq_a M_{k-1}(a)\succeq_a\cdots\succeq_a M_1(a)$ if $a\in L$. Since $v_a$ satisfies condition ($\star$), we have $v_a(M_i(a))\geq v_a(M_j(a))$ whenever $M_i(a)\succeq_a M_j(a)$. The claims follow.
\end{proof}

\subsection{Algorithm Description}

Utilizing the properties shown in the previous subsection, we provide an efficient algorithm for finding a stable matching with minimum utility gap.

For each agent $a\in A$, let
$\C(a)=\rho^a_1\lhd \rho^a_2\lhd\cdots\lhd \rho^a_{k_a}$
be the chain of rotations involving $a$, where $k_a=|\C(a)|$.
For notational convenience, we introduce a dummy element $\rho^a_0$, which is not a rotation, and extend the notation $\hat v_a$ by setting $\hat v_a(\rho^a_0)=v_a(M_0(a))$.
The proof of Lemma~\ref{lem:value-monotonicity} also establishes monotonicity after prepending this initial value.
Thus, if $|X\cap \C(a)|=i$, then $v_a(M_X(a))=\hat v_a(\rho^a_i)$.

Our main algorithm, given as Algorithm~\ref{alg:minimum-gap}, repeatedly calls Algorithm~\ref{alg:interval-feasibility} as a subroutine, which tests the feasibility of a given interval.
More precisely, for a given $[\Blow,\Bhigh]$, it decides whether there exists $M\in \S(I)$ such that
$\Blow\leq v_a(M(a))\leq \Bhigh$ for every $a\in A$.

\begin{algorithm}[t]
\caption{Interval Feasibility (A subroutine used in Algorithm~\ref{alg:minimum-gap})}
\label{alg:interval-feasibility}
\SetKwInOut{Input}{Input}
\SetKwInOut{Output}{Output}
\Input{A rotation digraph $G$ of $I$, chains $\C(a)=\rho^a_1\lhd\cdots\lhd\rho^a_{k_a}$ and values $\hat v_a(\rho^a_i)$ for all $a\in A$ and $0\leq i\leq k_a$, and bounds $\Blow,\Bhigh$.}
\Output{Either a closed set $X\subseteq\R(I)$ corresponding to a stable matching $M_X$ with $\Blow\leq v_a(M_X(a))\leq \Bhigh$ for every $a\in A$, or {\sc No}.}
\BlankLine
$F\gets\emptyset$ and $D\gets\emptyset$\tcp*{{\small $F$:\,forced rotations; $D$:\,disallowed rotations}}
\ForEach{$a\in A$}{
    Let $I_a=\set{i\in\{0,1,\ldots,k_a\}\mid \Blow\leq \hat v_a(\rho^a_i)\leq \Bhigh}$\;\label{line:interval}
    \If{$I_a=\emptyset$}{
        \Return {\sc No}\;
    }
    $\alpha\gets \min I_a$ and $\beta\gets \max I_a$\;
    \If{$\alpha>0$}{
        $F\gets F\cup\{\rho^a_{\alpha}\}$\;
    }
    \If{$\beta<k_a$}{
        $D\gets D\cup\{\rho^a_{\beta+1}\}$\;
    }
}
Compute the downward closure $\cl(F)=\set{\tau\in\R(I)\mid \tau\unlhd\rho \text{ for some } \rho\in F}$ of $F$ by a search in the reverse direction of the rotation digraph $G$\;
\If{$\cl(F)\cap D\neq\emptyset$}{
    \Return {\sc No}\;
}
\Return $\cl(F)$\;
\end{algorithm}

\begin{lemma}
\label{lem:interval-feasibility}
For any given interval $[\Blow,\Bhigh]$, Algorithm~\ref{alg:interval-feasibility} correctly determines its feasibility in $O(n^2)$ time.
\end{lemma}

\begin{proof}
By Lemma~\ref{lem:value-monotonicity}, for any agent $a\in A$, the sequence
$\hat v_a(\rho^a_0),\hat v_a(\rho^a_1),\ldots,\hat v_a(\rho^a_{k_a})$
is monotone non-increasing if $a\in S$ and monotone non-decreasing if $a\in L$.
Hence, for a fixed interval $[\Blow,\Bhigh]$, the set $I_a$ defined in line~\ref{line:interval} is consecutive; that is, whenever $I_a\neq\emptyset$, it can be written as $I_a=\{\alpha,\alpha+1,\ldots,\beta\}$ for some $\alpha\leq\beta$.

A closed set $X$ satisfies the utility bound for agent $a$ if and only if $\alpha\leq |X\cap \C(a)|\leq\beta$.
Since $X\cap \C(a)$ is a prefix of $\C(a)$, the condition $|X\cap \C(a)|\geq\alpha$ is equivalent to requiring $\rho^a_{\alpha}\in X$ when $\alpha>0$, while the condition $|X\cap \C(a)|\leq\beta$ is equivalent to requiring $\rho^a_{\beta+1}\notin X$ when $\beta<k_a$.

Thus, the set $F$ in the algorithm consists of rotations that must be contained in any feasible closed set, and the set $D$ consists of rotations that must be avoided.
Since any closed set containing $F$ must also contain its downward closure $\cl(F)$, a feasible closed set exists only if $\cl(F)\cap D=\emptyset$.
Conversely, if $\cl(F)\cap D=\emptyset$, then $X=\cl(F)$ is itself a closed set satisfying all the required bounds.
Therefore Algorithm~\ref{alg:interval-feasibility} is correct.

We finally evaluate the time complexity.
The algorithm scans the chains for agents in $A$ and computes the downward closure of $F$ in the rotation digraph.
Since the total length of the chains is $O(n^2)$ by Corollary~\ref{cor:agent-chain-length} and the rotation digraph has $O(n^2)$ arcs by Theorem~\ref{thm:rotation-digraph-construction}, these operations can be done in $O(n^2)$ time.
The remaining operations, including the construction of $F$ and $D$, also take $O(n^2)$ time.
Hence the algorithm runs in $O(n^2)$ time.
\end{proof}

Recall that we aim to minimize one of the following two objective functions:
\[
\Delta_{\rm diff}(M)=\max_{a\in A} v_a(M(a))-\min_{a\in A} v_a(M(a))
\quad\text{and}\quad
\Delta_{\rm rat}(M)= \frac{\max_{a\in A} v_a(M(a))}{\min_{a\in A} v_a(M(a))}.
\]
A stable matching $M$ that is feasible with respect to an interval $[\Blow,\Bhigh]$ has objective value at most $\Bhigh-\Blow$ in the $\Delta_{\rm diff}$ case and at most $\Bhigh/\Blow$ in the $\Delta_{\rm rat}$ case when $\Blow>0$, and these bounds are tight if this interval is an inclusion-wise minimal feasible interval. We then search for a minimal feasible interval with minimum objective value.

Let
$\mathcal{V}=\set{\hat v_a(\rho^a_i)\mid a\in A,\ 0\leq i\leq k_a}$
be the set of all utility values that can appear for agents in $A$, and let $b_1<b_2<\cdots<b_N$ be the sorted list of distinct values in $\mathcal{V}$.
By Corollary~\ref{cor:agent-chain-length}, we have $k_a\leq n$ for every $a\in A$, and hence $N=O(n^2)$.
The naive approach is to check the feasibility of intervals $[b_i,b_j]$ for all pairs $b_i\leq b_j$ and take one minimizing the objective value.
Since the number of such pairs is $O(N^2)=O(n^4)$, this already gives a polynomial-time algorithm.
However, we can speed this up by using the fact that feasibility is monotone with respect to interval inclusion: if an interval $[b_i,b_j]$ is infeasible, then every interval contained in it is also infeasible.
Thanks to this property, we can employ a sliding-window strategy \cite{cormen2022introduction} to reduce the number of intervals to be checked to $O(n^2)$.

Algorithm~\ref{alg:minimum-gap} uses a function $\gap$.
In the $\Delta_{\rm diff}$ version, it is defined as $\gap(b_i,b_j)=b_j-b_i$.
In the $\Delta_{\rm rat}$ version, assuming all values are positive, it is defined as $\gap(b_i,b_j)=b_j/b_i$.

\begin{algorithm}[t]
\caption{Stable Matching with Minimum Utility Gap (Main Algorithm)}
\label{alg:minimum-gap}
\SetKwInOut{Input}{Input}
\SetKwInOut{Output}{Output}
\Input{An instance $I$, an agent subset $A$, value functions $\{v_a\}_{a\in A}$, and the choice of objective function, either difference or ratio.}
\Output{A stable matching minimizing the selected utility-gap objective.}
\BlankLine
Construct the student-optimal stable matching $M_0$ in $I$, the set $\R(I)$ of rotations, and a rotation digraph $G=(\R(I),A_G)$\;
Construct the chain $\C(a)$ for each $a\in A$ and compute $\hat v_a(\rho^a_i)$ for all $i\in\{0,1,\dots,k_a\}$\;
Let $b_1<b_2<\cdots<b_N$ be the sorted list of distinct values in $\set{\hat v_a(\rho^a_i)\mid a\in A,\ 0\leq i\leq k_a}$\;
\smallskip
$\best\gets +\infty$, $X^\star\gets\emptyset$, $i\gets 1$, and $j\gets 1$\tcp*{{\small $\Delta$:\,upper bound on the optimal value}}
\While{$1\leq i\leq j\leq N$}{
    Execute Algorithm~\ref{alg:interval-feasibility} for the interval $[b_i,b_j]$\;
    \If{Algorithm~\ref{alg:interval-feasibility} returns {\sc No}}{
        $j\gets j+1$\;
    }
    \Else{
        \If(\tcp*[f]{{\small $\gap(b_i,b_j)=b_j-b_i$ for $\Delta_{\rm diff}\!$ and $\frac{b_j}{b_i}$ for $\Delta_{\rm rat}\!$}})
        {$\gap(b_i,b_j)<\best$}{
        \smallskip
            Let $X$ be the closed set returned by Algorithm~\ref{alg:interval-feasibility}\;
            $\best\gets \gap(b_i,b_j)$ and $X^\star\gets X$\;
        }
        $i\gets i+1$\;
    }
}
\Return the stable matching $M_{X^\star}$ corresponding to $X^\star$\;
\end{algorithm}

\begin{theorem}
\label{thm:main}
Each of the $\Delta_{\rm diff}$ and $\Delta_{\rm rat}$ versions of the Stable Matching Problem with Minimum Utility Gap can be solved in $O(n^4+n^2T_v)$ time, where $n=\max\{|S|,|L|\}$ and $T_v$ is the time required to evaluate one value $v_a(X)$. In particular, when $T_v=O(n^2)$, the algorithm runs in $O(n^4)$.
\end{theorem}

\begin{proof}
We can construct the student-optimal stable matching, the set of rotations, and a rotation digraph in $O(n^2)$ time by Theorem~\ref{thm:rotation-digraph-construction}.
By Corollary~\ref{cor:agent-chain-length}, each chain $\C(a)$ has length at most $n$, and hence the total number of values $\hat v_a(\rho^a_i)$ over all $a\in A$ and $0\leq i\leq k_a$ is $O(n^2)$.
These values can be computed using $O(n^2)$ value-oracle calls.
Therefore, the set $\mathcal V$ has size $N=O(n^2)$, and sorting it takes $O(n^2\log n)$ time.

It remains to justify the sliding-window search.
For each $i\in\{1,\dots, N\}$, let $J(i)$ be the minimum index $j\geq i$ such that the interval $[b_i,b_j]$ is feasible. If no such index exists, let $J(i)$ be $N+1$ for convenience.
Then, $J(i)$ is monotone non-decreasing in $i$. Indeed, if $i<i'$ and $[b_{i'},b_j]$ is feasible for some $j<J(i)$, then $[b_i,b_j]$ is also feasible by monotonicity with respect to interval inclusion, contradicting the minimality of $J(i)$.

Algorithm~\ref{alg:minimum-gap} maintains the upper pointer $j$ and moves it only forward.
For a fixed lower endpoint $b_i$, if the current interval $[b_i,b_j]$ is infeasible, the algorithm increases $j$ until it either reaches the minimum feasible upper endpoint $J(i)$ or exceeds $N$.
If it reaches $J(i)$, the algorithm records the value of this interval and then increases $i$.
By the monotonicity of $J(i)$, there is no need to move $j$ backward for later lower endpoints.
If the algorithm terminates because $i>j$, then this happens immediately after a feasible singleton interval $[b_j,b_j]$ has been found, which is clearly optimal.
If the algorithm terminates because $j>N$, then the current and all later lower endpoints admit no feasible interval, so no unchecked interval can improve the solution.

Thus, for every relevant lower endpoint $b_i$, the algorithm finds the feasible interval with the smallest possible upper endpoint.
For any feasible interval $[b_i,b_j]$, replacing $b_j$ by the smallest feasible upper endpoint cannot increase the value of $\gap$ in either the difference or the ratio version.
Therefore, a feasible interval $[b_i,b_j]$ with the minimum value of $\gap(b_i, b_j)$ must be found by the algorithm. For such an interval and the returned closed set $X$, $\gap(b_i, b_j)$ is exactly the optimal value; the feasibility implies $b_i\leq v_a(M_X(a))\leq b_j$ for every $a\in A$, and the minimality of $\gap(b_i, b_j)$ implies the existence of agents in $A$ attaining values $b_i$ and $b_j$.

Over the whole execution of Algorithm~\ref{alg:minimum-gap}, the pointer $i$ moves at most $N$ times and the pointer $j$ also moves at most $N$ times.
Therefore Algorithm~\ref{alg:interval-feasibility} is called $O(N)=O(n^2)$ times.
Since each feasibility check takes $O(n^2)$ time by Lemma~\ref{lem:interval-feasibility}, the total running time of the sliding-window search is $O(n^4)$.
Therefore the overall running time is $O(n^4+n^2T_v)$.
\end{proof}

When the value functions $v_a$ are total utility or average value functions defined by weights (see Section~\ref{sec:problem-setting}), we have $T_v=O(n)$, and hence the total running time is $O(n^4)$.

\section{Comparison with Existing Frameworks}
\label{sec:comparison}

In this section, we position our result relative to existing frameworks.

We first show that our objectives cannot be captured by the minimum-cut representability framework of Faenza--Foussoul--He \cite{FaenzaFoussoulHe2025}, and then explain that, when $A\subseteq S$ or $A\subseteq L$, our problem can be viewed as submodular function minimization over a ring family. We also note that this submodular interpretation does not extend to the general case where $A$ contains agents from both sides.

\subsection{Nonrepresentability by Minimum Cuts}
\label{subsec:mincut-nonrep}

Minimum-cut representability is a recent general framework that captures various optimization problems over stable matchings. Roughly speaking, a stable matching optimization problem is {\em minimum-cut representable} if it can be transformed into a minimum $s$--$t$ cut problem on a digraph whose vertices, other than $s$ and $t$, correspond to rotations. Faenza--Foussoul--He \cite{FaenzaFoussoulHe2025} gave a necessary and sufficient condition for such a representation.

We recall only the notation needed to explain that our problems are not included in the framework. In \cite{FaenzaFoussoulHe2025}, the feasibility set can be any sublattice $\mathcal F$ of the stable matching lattice; to explain our counterexample, it suffices to consider the special case $\mathcal F=\S(I)$. Following \cite{FaenzaFoussoulHe2025}, we consider the many-to-one setting, i.e., $q_s=1$ for every student $s\in S$.

For a stable matching $M\in\S(I)$, let $\Theta_M\subseteq\R(I)$ be the closed set of rotations corresponding to $M$. For a rotation $\theta\in\R(I)$, let $M^\theta$ and $M_\theta$ be the stable matchings corresponding to the downward closure $\operatorname{cl}(\theta)$ and  $\operatorname{cl}(\theta)\setminus\{\theta\}$, respectively. The first-order differential is
\[\partial f_\theta=f(M^\theta)-f(M_\theta),\]
and, for two distinct rotations $\theta,\theta'\in\R(I)$, the second-order differential is
\[
\partial^2 f_{\theta,\theta'}=
f(M^\theta\wedge M_{\theta'})
+f(M_\theta\wedge M^{\theta'})
-f(M_\theta\wedge M_{\theta'})
-f(M^\theta\wedge M^{\theta'}).
\]
The characterization theorem in \cite{FaenzaFoussoulHe2025}, specialized to $\mathcal F=\S(I)$, states that the minimization of $f:\S(I)\to\mathbb R$ is minimum-cut representable if and only if the following conditions hold:
\medskip
\begin{description}
    \item[\sf (a)] $\partial^2 f_{\theta,\theta'}\geq 0$ for every pair of distinct rotations $\theta,\theta'\in\R(I)$.
    \item[\sf (b)] For every $M\in\S(I)$, the objective value coincides with its second-order expansion
    \[
    f(M)=
    f(M_0)+
    \sum_{\theta\in\Theta_M}\partial f_\theta
    -
    \frac{1}{2}
    \sum_{\theta,\theta'\in\Theta_M:\,\theta\neq\theta'}
    \partial^2 f_{\theta,\theta'}.
    \]
\end{description}
We denote the right-hand side of this expansion by $f^{\sf apx}(M)$.

As shown in Example~\ref{ex:minimumcut} in Appendix~\ref{app:omitted}, there is an instance demonstrating that neither $\Delta_{\rm diff}$ nor $\Delta_{\rm rat}$ is minimum-cut representable, even in the one-to-one matching case with $A=L$.
This shows that our polynomial-time algorithm is not a direct consequence of the minimum-cut representability framework.

\subsection{Relation to Submodular Function Minimization}
\label{subsec:submodular-comparison}

As we have seen, our problems are not minimum-cut representable. Since the minimum-cut problem is a special case of submodular function minimization \cite{IwataFleischerFujishige2001,Schrijver2000}, it is natural to ask whether our objectives can be explained through the broader concept of submodularity. We show that, when $A\subseteq S$ or $A\subseteq L$, our problem can indeed be viewed as submodular function minimization, whereas this is not the case for a general mixed-side set $A$.

We prove the claim for the case $A\subseteq L$; the case $A\subseteq S$ is analogous. Let $\mathcal Q$ be the family of all closed subsets of the rotation poset $(\R(I),\unlhd)$. Then $\mathcal Q$ is a {\em ring family}, i.e., $X,Y\in \mathcal Q$ implies $X\cup Y,X\cap Y\in \mathcal Q$. By Theorem~\ref{thm:isomorphism}, for any $X\in \mathcal Q$, eliminating all rotations in $X$ from $M_0$ yields a stable matching, which we denote by $M_X$. Define two functions $f_{\rm diff}, f_{\rm rat}:\mathcal Q\to \mathbb{R}$ by $f_{\rm diff}(X)\coloneqq \Delta_{\rm diff}(M_X)$ and $f_{\rm rat}(X)\coloneqq \Delta_{\rm rat}(M_X)$. We show that both functions are submodular, i.e., they satisfy $f(X)+f(Y)\geq f(X\cup Y)+f(X\cap Y)$ for any $X,Y\in \mathcal Q$.

For each $a\in A$, define $v^*_a:\mathcal Q\to \mathbb{R}$ by
$v^*_a(X)=v_a(M_X(a))$ for every $X\in \mathcal Q$.
Recall the definition of $\hat{v}_a$ just before
Lemma~\ref{lem:value-monotonicity}; we extend $\hat v_a$ to a dummy
element $\rho^a_0$ by setting $\hat v_a(\rho^a_0)=v_a(M_0(a))$, where
$\rho^a_0\lhd \rho$ for every $\rho\in\mathcal C(a)$.
Then
\[
v^*_a(X)=v_a(M_X(a))=\hat{v}_a(\rho),
\quad\text{where $\rho=\max_{\unlhd}\bigl((X\cap\mathcal C(a))\cup\{\rho^a_0\}\bigr)$}.
\]
Thus $v^*_a(X)$ depends only on the largest element in the prefix
$X\cap\mathcal C(a)$ of the chain $\mathcal C(a)$, with the dummy element
$\rho^a_0$ used when this prefix is empty.
Since $\hat v_a$ is monotone non-decreasing on $\C(a)$ for $a\in A\subseteq L$ by Lemma~\ref{lem:value-monotonicity}, for any $X,Y\in \mathcal Q$, we have
\begin{align*}
v^*_a(X\cup Y)=\max\{ v^*_a(X), v^*_a(Y)\}\text{~~~and~~~}v^*_a(X\cap Y)=\min\{ v^*_a(X), v^*_a(Y)\}.
\end{align*}
Let $U_X=\max_{a\in A}v^*_a(X)$ and $u_X=\min_{a\in A}v^*_a(X)$, and define $U_Y$ and $u_Y$ similarly. Then
\begin{align}
\max_{a\in A}v^*_a(X\cup Y)
&= \max_{a\in A}\max\{ v^*_a(X), v^*_a(Y)\}
= \max\{U_X,U_Y\},\label{eq:union-max}\\
\min_{a\in A}v^*_a(X\cup Y)
&= \min_{a\in A}\max\{ v^*_a(X), v^*_a(Y)\}
\geq \max\{u_X,u_Y\},\label{eq:union-min}\\
\max_{a\in A}v^*_a(X\cap Y)
&= \max_{a\in A}\min\{ v^*_a(X), v^*_a(Y)\}
\leq \min\{U_X,U_Y\},\label{eq:inter-max}\\
\min_{a\in A}v^*_a(X\cap Y)
&= \min_{a\in A}\min\{ v^*_a(X), v^*_a(Y)\}
= \min\{u_X,u_Y\}.\label{eq:inter-min}
\end{align}
Here, \eqref{eq:union-min} follows from $\max\{ v^*_a(X), v^*_a(Y)\}\geq \max\{u_X,u_Y\}$ for every $a\in A$, and \eqref{eq:inter-max} follows from $\min\{ v^*_a(X), v^*_a(Y)\}\leq \min\{U_X,U_Y\}$ for every $a\in A$.

For the difference version, using \eqref{eq:union-max}--\eqref{eq:inter-min}, we obtain
\begin{align*}
f_{\rm diff}(X)+f_{\rm diff}(Y)
&= (U_X-u_X)+(U_Y-u_Y)
\\
&= \max\{U_X,U_Y\}+\min\{U_X,U_Y\}-\max\{u_X,u_Y\}-\min\{u_X,u_Y\}\\
&\geq \max_{a\in A}v^*_a(X\cup Y)+\max_{a\in A}v^*_a(X\cap Y)-\min_{a\in A}v^*_a(X\cup Y)-\min_{a\in A}v^*_a(X\cap Y)\\
&=f_{\rm diff}(X\cup Y)+f_{\rm diff}(X\cap Y).
\end{align*}

For the ratio version, assume that all values are positive. We use a general inequality $\frac{x}{z}+\frac{y}{w}\geq \frac{\max\{x,y\}}{\max\{z,w\}}+\frac{\min\{x,y\}}{\min\{z,w\}}$, which holds for any positive real numbers $x,y,z,w$.\footnote{A brief proof. If $(x-y)(z-w)\geq 0$, then the two sides are equal. Otherwise, the right-hand side is $\frac{x}{w}+\frac{y}{z}$, and the difference between the left-hand side and the right-hand side is $\frac{x}{z}+\frac{y}{w}-\frac{x}{w}-\frac{y}{z}=\frac{(x-y)(w-z)}{zw}\geq 0$.}
Using this inequality and \eqref{eq:union-max}--\eqref{eq:inter-min}, we obtain
\begin{align*}
f_{\rm rat}(X)+f_{\rm rat}(Y)
&= \frac{U_X}{u_X}+\frac{U_Y}{u_Y}\\
&\geq \frac{\max\{U_X,U_Y\}}{\max\{u_X,u_Y\}}+\frac{\min\{U_X,U_Y\}}{\min\{u_X,u_Y\}}\\
&\geq \frac{\max_{a\in A}v^*_a(X\cup Y)}{\min_{a\in A}v^*_a(X\cup Y)}+\frac{\max_{a\in A}v^*_a(X\cap Y)}{\min_{a\in A}v^*_a(X\cap Y)}\\
&=f_{\rm rat}(X\cup Y)+f_{\rm rat}(X\cap Y).
\end{align*}

Thus both $f_{\rm diff}$ and $f_{\rm rat}$ are submodular functions on $\mathcal Q$. Hence, when $A$ lies on one side, a minimizer $X^*$ of the corresponding submodular function yields a stable matching $M_{X^*}$ minimizing the objective value. This view does not improve the running time of our algorithm, but it clarifies the position of our framework.

The proof above relies on the identities $v^*_a(X\cup Y)=\max\{ v^*_a(X), v^*_a(Y)\}$ and $v^*_a(X\cap Y)=\min\{ v^*_a(X), v^*_a(Y)\}$, which hold for every $a\in A\subseteq L$. For a student $a\in S$, the same identities hold with $\max$ and $\min$ interchanged. Therefore, when $A$ contains agents from both $S$ and $L$, the above argument breaks down. Indeed, there are instances in which the resulting objective function is not submodular. See Example~\ref{ex:nonsubmodular} in Appendix~\ref{app:omitted}.

\section{Conclusion}
We introduced the Stable Matching Problem with Minimum Utility Gap and gave an efficient algorithm for solving it. The algorithm exploits the chain structure of rotations involving a fixed agent in the rotation poset. An interesting future direction is to investigate how far our tractability result extends beyond the present setting, e.g., to models with choice functions admitting compact poset representations.
Another interesting direction is to characterize the nonlinear objectives that can be optimized by the same approach as in the present paper, by exploiting the agent-wise chain structure.

\section*{Acknowledgements}
The authors thank Yuri Faenza and Satoru Iwata for their valuable comments, and the anonymous reviewers for their helpful feedback.
The second author was supported by JST ERATO Grant Number JPMJER2301, JST CRONOS Grant Number JPMJCS24K2, and JSPS KAKENHI Grant Number JP26K14706.

\bibliography{bib2doi}
\appendix

\section{Omitted Proofs and Examples}
\label{app:omitted}

In this section, we provide proofs and examples omitted from the main body due to space constraints.

We first prove that eliminating an exposed rotation preserves stability, as stated without proof in Section~\ref{sec:preliminaries}.

\begin{lemma}
\label{lem:rotation-stability}
Suppose that a rotation $\rho=((s_0,\ell_0),\ldots,(s_{t-1},\ell_{t-1}))$ is exposed in a stable matching $M$. Then $M/\rho$ is a stable matching.
\end{lemma}
\begin{proof}
By the definition of a rotation, each participating laboratory $\ell_i$ is full and replaces its
worst assignee with a preferred student. Thus, removed pairs cannot
block $M/\rho$, and no laboratory becomes more willing to accept a student.
By the stability of $M$, each participating student $s_i$ is full
and prefers every laboratory in $M(s_i)$ to the new partner
$\ell_{i-1}$.
Any blocking pair of $M/\rho$ must therefore involve a participating
student $s_i$ and a laboratory $\ell\notin M(s_i)$ that was already
willing to accept $s_i$ in $M$ and satisfies
$\ell\succ_{s_i}\ell_{i-1}$. This contradicts the last condition in the definition of rotations (Definition~\ref{def:rotation}).
\end{proof}

The following lemma is given in Section~\ref{sec:algorithm}.
\uniquenesslem*

\begin{proof}
Statement (i) is exactly Lemma~6 of Bansal--Agrawal--Malhotra \cite{BansalAgrawalMalhotra2007}.

We prove (iii). Suppose, to the contrary, that there exist a rotation $\rho$ and a pair $(s,\ell)$ such that $(s,\ell)\in M_z\cap\rho^-$. Let $M$ be a stable matching in which $\rho$ is exposed. Then $(s,\ell)\notin M/\rho$, while $M/\rho$ contains $(s,\rho^+(s))$. By Lemma~\ref{lem:one-step-monotonicity}, $\ell=\rho^-(s)\succ_s\rho^+(s)$. By Theorem~\ref{thm:comparability}, $M_z(s)$ and $(M/\rho)(s)$ are comparable. However, $(M/\rho)(s)\succeq_s M_z(s)$ is impossible, because $\rho^+(s)\in (M/\rho)(s)$ and $\ell\in M_z(s)\setminus (M/\rho)(s)$, but $\rho^+(s)\not\succ_s\ell$. Hence $M_z(s)\succ_s (M/\rho)(s)$. This contradicts the fact that $M_z$ is student-pessimal. Thus, $M_z\cap\rho^-=\emptyset$.

Similarly, suppose that there exist a rotation $\rho$ and a pair $(s,\ell)$ such that $(s,\ell)\in M_0\cap\rho^+$. Let $M$ be a stable matching in which $\rho$ is exposed. Then $(s,\ell)\notin M$, and $M(s)\succ_s\rho^+(s)=\ell$ by Lemma~\ref{lem:one-step-monotonicity}. Since $|M_0(s)|=|M(s)|$ by Theorem~\ref{thm:RHT}, there exists $h\in M(s)\setminus M_0(s)$, and it satisfies $h\succ_s\ell$. By Theorem~\ref{thm:comparability}, $M_0(s)$ and $M(s)$ are comparable. However, $M_0(s)\succeq_s M(s)$ is impossible, because $\ell\in M_0(s)$ and $h\in M(s)\setminus M_0(s)$, but $\ell\not\succ_s h$. Hence $M(s)\succ_s M_0(s)$. This contradicts the student-optimality of $M_0$. Thus, $M_0\cap\rho^+=\emptyset$.

It remains to prove (ii). By Theorem~\ref{thm:isomorphism}, $M_z$ is obtained from $M_0$ by eliminating all rotations exactly once. For an acceptable student-laboratory pair $e\in E$, let $r^-(e)$ be the number of rotations $\rho$ with $e\in\rho^-$, and let $r^+(e)$ be the number of rotations $\rho$ with $e\in\rho^+$. Consider a full elimination sequence from $M_0$ to $M_z$. Since each rotation deletes only existing elements and adds only nonexistent elements, for any $e\in E$, the value ${\bf 1}_{M_0}(e)-r^-(e)+r^+(e)$ is $1$ if $e\in M_z$ and $0$ otherwise, where ${\bf 1}_{M_0}(e)$ is the indicator of $e\in M_0$. By (i), $r^-(e)\leq 1$. By the first part of (iii), if $e\in M_z$, then $r^-(e)=0$; by the second part of (iii), if $e\in M_0$, then $r^+(e)=0$. Therefore, the required condition $r^+(e)\leq 1$ is immediate if $e\in M_0$. If $e\notin M_0$ and $e\in M_z$, then ${\bf 1}_{M_0}(e)-r^-(e)+r^+(e)=1$ and $r^-(e)=0$, and hence $r^+(e)=1$. If $e\notin M_0\cup M_z$, then ${\bf 1}_{M_0}(e)-r^-(e)+r^+(e)=0$ gives $r^+(e)=r^-(e)\leq 1$. Hence $r^+(e)\leq 1$ in all cases.
\end{proof}

The following examples are deferred from Section~\ref{sec:comparison}.

\begin{example}[Nonrepresentability by Minimum Cuts]\label{ex:minimumcut}
We show that both of our objective functions $\Delta_{\rm diff}$ and $\Delta_{\rm rat}$ violate condition~(b) of the characterization of minimum-cut representability in \cite{FaenzaFoussoulHe2025}, reviewed in Section~\ref{sec:comparison}. Recall that condition~(b) requires an objective function $f$ to satisfy
\[
f(M)=f^{\sf apx}(M)\coloneqq
f(M_0)+
\sum_{\theta\in\Theta_M}\partial f_\theta
-
\frac{1}{2}
\sum_{\theta,\theta'\in\Theta_M:\,\theta\neq\theta'}
\partial^2 f_{\theta,\theta'}
\]
for any $M\in \S(I)$ where
\begin{align*}
&\partial f_\theta=f(M^\theta)-f(M_\theta),\\
&\partial^2 f_{\theta,\theta'}=
f(M^\theta\wedge M_{\theta'})
+f(M_\theta\wedge M^{\theta'})
-f(M_\theta\wedge M_{\theta'})
-f(M^\theta\wedge M^{\theta'}).
\end{align*}

Consider the disjoint union of four identical $2\times 2$ one-to-one matching gadgets. For each $i\in\{1,2,3,4\}$, the $i$th gadget consists of students $s_i,s'_i$ and laboratories $\ell_i,\ell'_i$. The acceptable pairs are only those within the same gadget, and the preference lists are
\[
s_i:\ell_i\succ\ell'_i,\qquad
s'_i:\ell'_i\succ\ell_i,\qquad
\ell_i:s'_i\succ s_i,\qquad
\ell'_i:s_i\succ s'_i.
\]
The student-optimal stable matching $M_0$ assigns $s_i$ to $\ell_i$ and $s'_i$ to $\ell'_i$ for every $i$. Each gadget has exactly one rotation, denoted by $\theta_i$, and eliminating $\theta_i$ replaces $(s_i,\ell_i),(s'_i,\ell'_i)$ with $(s_i,\ell'_i),(s'_i,\ell_i)$. The rotations $\theta_1,\theta_2,\theta_3,\theta_4$ are mutually incomparable in $(\R(I),\unlhd)$.

Let $A=L$ and define the values of the laboratories by
\[
v_{\ell_i}(\{s_i\})=10,\quad
v_{\ell_i}(\{s'_i\})=20,\quad
v_{\ell'_i}(\{s'_i\})=10,\quad
v_{\ell'_i}(\{s_i\})=20
\]
for each $i\in\{1,2,3,4\}$.
For $T\subseteq\{1,2,3,4\}$, let $M_T$ be the stable matching obtained from $M_0$ by eliminating the rotations $\set{\theta_i \mid i\in T}$. Then
\[
\Delta_{\rm diff}(M_T)=
\begin{cases}
0 & \text{if } |T|=0 \text{ or } |T|=4,\\
10 & \text{if } 1\leq |T|\leq 3,
\end{cases}
\quad
\Delta_{\rm rat}(M_T)=
\begin{cases}
1 & \text{if } |T|=0 \text{ or } |T|=4,\\
2 & \text{if } 1\leq |T|\leq 3.
\end{cases}
\]

We first consider the difference version and let $f(M)=\Delta_{\rm diff}(M)$.
Since all rotations are incomparable, $M_{\theta_i}=M_0$ and $M^{\theta_i}=M_{\{i\}}$ for each $i$. Hence $\partial f_{\theta_i}=f(M_{\{i\}})-f(M_0)=10$. For distinct $i,j$, the meet corresponds to taking the union of eliminated rotations, and therefore
$\partial^2 f_{\theta_i,\theta_j}
=
f(M_{\{i\}})+f(M_{\{j\}})-f(M_0)-f(M_{\{i,j\}})
=
10$.
Now consider $M=M_{\{1,2,3,4\}}$. Then all laboratories have utility $20$, so $f(M)=0$. On the other hand, $\Theta_M=\{\theta_1,\theta_2,\theta_3,\theta_4\}$, and the second-order expansion gives
\[
f^{\sf apx}(M)
=
0+4\cdot 10-\frac{1}{2}\cdot 12\cdot 10
=
-20.
\]
Thus $f(M)\neq f^{\sf apx}(M)$, and condition (b) fails. Consequently, $\Delta_{\rm diff}$ is not minimum-cut representable in general.

Next, let $f(M)=\Delta_{\rm rat}(M)$. Then $f(M_0)=1$, $f(M_{\{i\}})=2$, $f(M_{\{i,j\}})=2$, and $f(M_{\{1,2,3,4\}})=1$. Hence $\partial f_{\theta_i}=1$ and $\partial^2 f_{\theta_i,\theta_j}=1$ for all distinct $i,j$, while
\[
f^{\sf apx}(M_{\{1,2,3,4\}})
=
1+4\cdot 1-\frac{1}{2}\cdot 12\cdot 1
=
-1
\neq
1
=
f(M_{\{1,2,3,4\}}).
\]
Thus condition (b) fails for $\Delta_{\rm rat}$ as well.
\end{example}

\begin{example}[Non-submodularity when $A$ intersects both $S$ and $L$]
\label{ex:nonsubmodular}
We show that, unlike the case with $A\subseteq S$ or $A\subseteq L$ analyzed in Section~\ref{subsec:submodular-comparison}, when $A\cap S\neq \emptyset$ and $A\cap L\neq \emptyset$, the objective functions $f_{\rm diff}, f_{\rm rat}:\mathcal Q\to \mathbb{R}$ defined by $f_{\rm diff}(X)\coloneqq \Delta_{\rm diff}(M_X)$ and $f_{\rm rat}(X)\coloneqq \Delta_{\rm rat}(M_X)$ are not necessarily submodular. Here, $\mathcal Q$ is the family of closed sets of the rotation poset, and $M_X$ is the stable matching corresponding to $X\in \mathcal Q$.

We consider an instance similar to Example~\ref{ex:minimumcut}, but here we use only two gadgets. For each $i\in \{1,2\}$, create a $2\times 2$ one-to-one matching gadget consisting of $s_i, s'_i, \ell_i, \ell'_i$ as in Example~\ref{ex:minimumcut}. Then the student-optimal stable matching $M_0$ assigns $s_i$ to $\ell_i$ and $s'_i$ to $\ell'_i$ for each $i\in\{1,2\}$. There are two rotations $\theta_1$ and $\theta_2$, and eliminating $\theta_i$ replaces $(s_i,\ell_i),(s'_i,\ell'_i)$ with $(s_i,\ell'_i),(s'_i,\ell_i)$. These two rotations are mutually incomparable in the rotation poset $(\R(I),\unlhd)$.

Let $A=S\cup L$; thus, $A$ consists of all eight agents.
For singleton sets, we write $v_a(x)$ instead of $v_a(\{x\})$.
The agents' values are defined as follows:
\begin{align}
v_{s_1}(\ell_1)=v_{s'_1}(\ell'_1)&=30, \qquad
v_{\ell_1}(s_1)=v_{\ell'_1}(s'_1)=10, \label{eq:values1}\\
v_{s_2}(\ell_2)=v_{s'_2}(\ell'_2)&=20,\qquad
v_{\ell_2}(s_2)=v_{\ell'_2}(s'_2)=20, \label{eq:values2}\\[1mm]
v_{s_1}(\ell'_1)=v_{s'_1}(\ell_1)&=20, \qquad
v_{\ell_1}(s'_1)=v_{\ell'_1}(s_1)=20, \label{eq:values3}\\
v_{s_2}(\ell'_2)=v_{s'_2}(\ell_2)&=10, \qquad
v_{\ell_2}(s'_2)=v_{\ell'_2}(s_2)=30. \label{eq:values4}
\end{align}
These values are consistent with the preferences; each agent assigns a larger value to a better partner. The values in \eqref{eq:values1} and \eqref{eq:values2} are attained in the student-optimal stable matching, while those in \eqref{eq:values3} and \eqref{eq:values4} are attained in the student-pessimal stable matching.

Set $X=\{\theta_1\}$ and $Y=\{\theta_2\}$. Then $X\cap Y=\emptyset$ and $X\cup Y=\{\theta_1,\theta_2\}$ correspond to the student-optimal and student-pessimal stable matchings, respectively. Thus, we have
\[
f_{\rm diff}(X\cup Y)=20,\quad f_{\rm diff}(X\cap Y)=20,\quad
f_{\rm rat}(X\cup Y)=3,\quad f_{\rm rat}(X\cap Y)=3.
\]
The set $X$ corresponds to the matching assigning the values in \eqref{eq:values2} and \eqref{eq:values3}, while the set $Y$ corresponds to the matching assigning the values in \eqref{eq:values1} and \eqref{eq:values4}. Thus,
\[
f_{\rm diff}(X)=0,\quad f_{\rm diff}(Y)=20,\quad
f_{\rm rat}(X)=1,\quad f_{\rm rat}(Y)=3.
\]
From these, we obtain
\begin{align*}
&f_{\rm diff}(X)+f_{\rm diff}(Y)=20<40=f_{\rm diff}(X\cup Y)+f_{\rm diff}(X\cap Y),\\
&f_{\rm rat}(X)+f_{\rm rat}(Y)=4<6=f_{\rm rat}(X\cup Y)+f_{\rm rat}(X\cap Y),
\end{align*}
and hence both $f_{\rm diff}$ and $f_{\rm rat}$ violate the submodular inequality.
\end{example}

\section{NP-hardness and Inapproximability with Ties}
\label{app:hardness}

In this section, we consider the case where ties are allowed in the preference lists of agents. We use the standard weak-stability notion under ties \cite{irving1994stable}: a pair blocks a matching only if both agents can become {\em strictly} better off. We show that, for both versions of the objective function, finding an optimal stable matching becomes NP-hard. We also show inapproximability results. The results in this section hold even in the many-to-one matching setting, even when $A=L$, ties appear only in laboratories' preference lists, and all pairs are acceptable, i.e., $E=S\times L$.

We first provide the NP-hardness proof for the difference version. In this case, simple average value functions defined by weights are sufficient to show the hardness; the proof can also be modified to work for total utility value functions.
\begin{theorem}
\label{thm:hardness-ties}
The difference version of the Stable Matching Problem with Minimum Utility Gap in the presence of ties is NP-hard.
\end{theorem}

\begin{proof}
We give a polynomial-time reduction from the \textsc{Partition} problem. An instance of \textsc{Partition} is a set $\{a_1,a_2,\ldots,a_m\}$ of positive integers with total sum $T=\sum_{i=1}^m a_i$, where $T$ is assumed to be even. The question is whether there exists a subset $J\subseteq\{1,\ldots,m\}$ such that $\sum_{i\in J}a_i=T/2$.

From this instance, we construct a many-to-one matching instance. For each $i\in\{1,\ldots,m\}$, we create a local laboratory $\ell_i$ with capacity $1$ and two students $s_i$ and $t_i$. We also create a special global laboratory $\ell_{\mathrm{sum}}$ with capacity $m$. Thus, we have $S=\{s_1,t_1, \dots, s_m, t_m\}$ and $L=\{\ell_1, \dots, \ell_m\}\cup \{\ell_{\rm sum}\}$. We set $A=L$.

Each student $s_i$ and $t_i$ ranks $\ell_i$ first, $\ell_{\mathrm{sum}}$ second, and all other laboratories afterward in an arbitrary strict order. That is, it is written as
\[s_i, t_i: \ell_i\succ \ell_{\rm sum}\succ \cdots.\]

For each $\ell\in L=A$, we define a weight function $w_\ell:S\to \mathbb{R}$, using which we define the value function $v_\ell$ by $v_\ell(X)=\frac{1}{|X|}\sum_{x\in X} w_\ell(x)$ for $X\neq\emptyset$, with $v_\ell(\emptyset)=0$. For each local laboratory $\ell_i$, set
\[\textstyle w_{\ell_i}(s_i)=w_{\ell_i}(t_i)=\frac{T}{2},\]
and assign weight $0$ to all other students. For the global laboratory $\ell_{\mathrm{sum}}$, set
\[w_{\ell_{\mathrm{sum}}}(s_i)=ma_i,\qquad w_{\ell_{\mathrm{sum}}}(t_i)=0\]
for each $i\in\{1,\ldots,m\}$.

For each $\ell\in L=A$, let its preference be consistent with $w_\ell$, where students with the same value are tied. Then the preference list of each local laboratory $\ell_i$ has the form
\[\ell_i: (s_i\sim t_i)\succ \cdots.\]
This completes the construction of a matching instance $I$.

In any stable matching in $I$, each local laboratory $\ell_i$ is assigned exactly one of $s_i$~and~$t_i$. Indeed, if $\ell_i$ were assigned neither of them, then one of these two students would form a blocking pair with $\ell_i$. Conversely, once $\ell_i$ is assigned one of them, the other does not form a blocking pair with $\ell_i$ because $\ell_i$ is indifferent between the two. This also implies that the number of students who are not assigned to their first-choice local laboratory is exactly $m$, and hence they are all assigned to their second choice $\ell_{\rm sum}$, whose capacity is $m$. Thus, for each $i\in \{1,\dots,m\}$, there are two possible stable states:
\begin{align*}
&\text{state 0: } M(s_i)=\ell_i,\ M(t_i)=\ell_{\mathrm{sum}},\\
&\text{state 1: } M(t_i)=\ell_i,\ M(s_i)=\ell_{\mathrm{sum}}.
\end{align*}
Conversely, any independent choice of these two states for each $i$ yields a stable matching: only the unassigned first-choice pairs could block, but the corresponding laboratories are indifferent between their current and proposed assignees.

Let $J=\set{i | M(s_i)=\ell_{\mathrm{sum}}}$. In either state, the utility received by each local laboratory $\ell_i$ is $v_{\ell_i}(M(\ell_i))=\frac{T}{2}$. On the other hand, the utility of the global laboratory is
\[\textstyle v_{\ell_{\mathrm{sum}}}(M(\ell_{\mathrm{sum}}))=\frac{1}{m}\sum_{i\in J}ma_i=\sum_{i\in J}a_i.\]
Hence the utility gap is
\[\textstyle \Delta_{\rm diff}(M)=\left|\sum_{i\in J}a_i-\frac{T}{2}\right|.\]
Therefore, the optimum value is $0$ if and only if the original \textsc{Partition} instance is a yes-instance. This proves NP-hardness.
\end{proof}

\begin{corollary}
\label{cor:hardness-inapprox}
For any finite $\alpha>1$ (possibly depending on the input size), there is no polynomial-time $\alpha$-approximation algorithm for the difference version of the Stable Matching Problem with Minimum Utility Gap in the presence of ties, unless $\mathrm{P}=\mathrm{NP}$.
\end{corollary}

\begin{proof}
Suppose that such an $\alpha$-approximation algorithm exists. In the reduction above, if the given \textsc{Partition} instance is a yes-instance, then the optimum value is $0$, so an $\alpha$-approximation algorithm must return a solution of value $0$. If the \textsc{Partition} instance is a no-instance, then every feasible solution has strictly positive value. Thus, the algorithm would distinguish yes-instances from no-instances of \textsc{Partition} in polynomial time, which is impossible unless $\mathrm{P}=\mathrm{NP}$.
\end{proof}

By slightly modifying the reduction used for the difference version, we can also show NP-hardness and inapproximability for the ratio version of our problem with ties.

Rather than providing a full proof again, we explain how to modify the construction in the reduction. Fix a positive integer $C$, independent of the input size, and define the increasing function
\[h_C(x)=\frac{1+C\max\{x-T/2,0\}}{1+C\max\{T/2-x,0\}}.\]
Keep the same weights and preference lists, but replace each laboratory's value function by
\[\textstyle v_\ell(X)=h_C\!\left(\frac{1}{|X|}\sum_{x\in X}w_\ell(x)\right)\quad\text{for }X\neq\emptyset,\]
with $v_\ell(\emptyset)=h_C(0)$. These functions are positive, preserve the consistency condition ($\star$) in Section~\ref{sec:problem-setting}, and are exactly evaluable in polynomial time using rational arithmetic. In any stable matching, each local laboratory has utility $h_C(T/2)=1$, while the global laboratory has utility $h_C(\sum_{i\in J}a_i)$. Hence
\[\textstyle \Delta_{\rm rat}(M)=1+C\cdot\left|\sum_{i\in J}a_i-\frac{T}{2}\right|.\]
Therefore, the optimum value is $1$ if and only if the original \textsc{Partition} instance is a yes-instance. This modification implies the following.

\begin{corollary}
The ratio version of the Stable Matching Problem with Minimum Utility Gap in the presence of ties is NP-hard.
\end{corollary}

By choosing the positive integer $C$ appropriately, we can also show the following inapproximability result. Note that the statement below for the ratio version is weaker than the corresponding statement, Corollary~\ref{cor:hardness-inapprox}, for the difference version because the former rules out constant-factor approximation algorithms, while the latter rules out even approximation algorithms whose factors may depend on the input size.

\begin{corollary}
For any fixed constant $\alpha>1$ independent of the input size, there is no polynomial-time $\alpha$-approximation algorithm for the ratio version of the Stable Matching Problem with Minimum Utility Gap in the presence of ties, unless $\mathrm{P}=\mathrm{NP}$.
\end{corollary}

\begin{proof}
Choose a fixed positive integer $C>\alpha-1$. If the given \textsc{Partition} instance is a yes-instance, the optimum ratio is $1$; otherwise, since $T/2$ and all $a_i$ are integers, it is at least $1+C>\alpha$. Thus, an $\alpha$-approximation algorithm would distinguish yes-instances from no-instances of \textsc{Partition} in polynomial time, which is impossible unless $\mathrm{P}=\mathrm{NP}$.
\end{proof}

\end{document}